\documentclass{article}
\usepackage{amsfonts,color,framed,amsthm,amsmath,hyperref,
mathtools,authblk,tikz,mathdots}

\usetikzlibrary{snakes}
\numberwithin{equation}{section}
\allowdisplaybreaks[1]

\def\C{\ensuremath{\mathbb C}}

\newcommand{\be}{\begin{equation}}
\newcommand{\ee}{\end{equation}}
\newcommand{\bea}{\begin{eqnarray}}
\newcommand{\eea}{\end{eqnarray}}

\newcommand{\order}[1]{\ensuremath{{\mathcal O}\left(#1\right)}}

\newcommand{\bt}{{\mathbf t}}
\newcommand{\bx}{{\mathbf x}}

\newcommand{\Span}{\operatorname{Span}}
\newcommand{\M}{\mathcal{M}}
\newcommand{\e}{\epsilon}
\newcommand{\p}{\partial}
\newcommand{\nn}{\nonumber}
\newcommand{\bgT}{{\mathbf T}}
\newcommand{\btheta}{{\boldsymbol \theta}}
\newtheorem{theorem}{Theorem}[section]
\newtheorem{definition}{Definition}[section]
\newtheorem{proposition}{Proposition}[section]
\newtheorem{lemma}{Lemma}[section]

\newtheorem{remark}{Remark}[section]
\newtheorem{example}{Example}[section]
\newenvironment{prf}{\noindent {\it Proof} \ }{\hfill $\Box$}

\begin{document}
\title{Geometric interpretation of Zhou's explicit formula for the Witten--Kontsevich tau function}
\author{Ferenc Balogh\footnote{fbalogh@sissa.it},~ Di Yang\footnote{dyang@sissa.it}\\
{\small SISSA, via Bonomea 265, Trieste 34136, Italy}}
\date{}
\maketitle
\begin{abstract}
Based on the work of Itzykson and Zuber on Kontsevich's integrals, we give a geometric interpretation and a simple proof of Zhou's explicit formula for the Witten-Kontsevich tau function.  More precisely, we show that the numbers $A_{m,n}^{Zhou}$ defined by Zhou coincide with the {\em affine coordinates} for the point of the Sato Grassmannian corresponding to the Witten-Kontsevich tau function.  Generating functions and new recursion relations for $A_{m,n}^{Zhou}$ are derived. Our formulation on matrix-valued affine coordinates and on tau functions remains valid for generic Grassmannian solutions of the KdV hierarchy. 

A by-product of our study indicates an interesting relation between the {\em matrix-valued} affine coordinates for the Witten-Kontsevich tau function and the $V$-matrices associated to the $R$-matrix of Witten's $3$-spin structures.
\end{abstract}

\section{Introduction}
The intersection theory of tautological classes over  the Deligne--Mumford moduli spaces  $\overline{\M}_{g,n}$ is an important subject in algebraic geometry and string theory. 

Recall that $\overline{\M}_{g,n}$ are moduli spaces of stable curves of genus $g$ with $n$ marked points. 
Denote by $\mathcal{L}_i$ the $i_{\text{th}}$ tautological line bundle over $\overline{\M}_{g,n},\,i=1,...,n$. 
The intersection numbers on $\overline{\M}_{g,n}$ are certain rational numbers defined by
\be \langle \tau_{k_1}...\tau_{k_n}\rangle_g := \int_{\overline{\M}_{g,n}} \psi_1^{k_1}...\psi_n^{k_n}\ee
where $k_1,...,k_n$ are non-negative integers and $\psi_i$ are first Chern classes of the line bundles $\mathcal{L}_i$
\be \psi_i:=c_1(\mathcal{L}_i).\ee
For example, $\langle \tau_0^3\rangle_0 =1,\,\langle \tau_1\rangle_1=\frac{1}{24}.$ Due to dimension counting, $\langle \tau_{k_1}...\tau_{k_n}\rangle_g$ gives zero unless
\be \label{dim}
k_1+...+k_n=3g-3+n.
\ee
The generating function of these numbers, also called the partition function of 2D quantum gravity, is the following formal series 
\be Z(\bt; \e)=\exp\left(\sum_{g=0}^\infty \e^{2g-2} \sum_{n=0}^\infty \sum_{k_1\geq0,...,k_n\geq 0} \frac{1}{n!} \langle \tau_{k_1}\cdots\tau_{k_n}\rangle_g\,  t_{k_1}\cdots t_{k_n}\right).\ee
Here $\bt=(t_0,t_1,t_2,\dots)$ are called coupling constants and $\e$ is the string coupling constant.

In 1991 E. Witten \cite{Witten} proposed his famous conjecture, which opened a new direction in the studies of the intersection theory, namely the relations between generating functions of intersection numbers and integrable hierarchies.

\paragraph{Witten's conjecture \cite{Witten}.}\ 
The function $u(\bt;\e)$ defined by
\be
u(\bt;\e):=\e^2 \p_x^2 \log Z(\bt;\e)
\ee
satisfies the Korteweg-de Vries (KdV) hierarchy:
\bea
u_{t_1}&=&u u_x+\frac{\e^2}{12}u_{xxx},\label{KdV-hie1}\\
u_{t_p}&=&\frac{1}{1+2p}\mathcal{R}\frac{\p u}{\p t_{p-1}},\quad p\geq 2,\label{KdV-hie2}
\eea
where $\mathcal{R}=2u+u_x\p_x^{-1}+\frac{\e^2}{4}\p^2$ is the Lenard-Magri recursion operator and $x$ is identified with $t_0$. Note that the solution $u(\bt;\e)$ is uniquely determined in the formal power series ring by the initial condition
\be u(t_0=x,0,0,...;\e)=x.\ee 
It was proved by Witten that the partition function $Z(\bt;\e)$ satisfies the string equation
\be
\sum_{p\geq 1} t_p \frac{\p Z}{\p t_{p-1}}+\frac{t_0^2}{2\e^2} Z=\frac{\p Z}{\p t_0}.
\ee
Together with the KdV hierarchy and the dimension counting \eqref{dim}, the string equation determines $Z(\bt;\e)$ uniquely.

Witten's conjecture was proved by Kontsevich \cite{Kontsevich}, and the partition function $Z(\bt;\e)$ is now famously known as the {\em Witten-Kontsevich  tau function}.  Below, for simplicity, we always take 
\be
\epsilon=1
\ee
and denote $Z(\bt)=Z(\bt;\e=1).$

One of the main tools in studying the KdV hierarchy is Sato's infinite dimensional Grassmannian \cite{Sato}. Kac and Schwarz \cite{KS} characterized the point of the Sato Grassmannian corresponding to the Witten-Kontsevich tau function as the linear subspace of $\mathbb{C}[[\lambda^{-1}]]$ given by
\be
W^{WK}=\Span_\mathbb{C} \{c(\lambda),\,S_\lambda c(\lambda), \, S_\lambda^2 c(\lambda),...\},
\ee
where $S_\lambda$ is the differential operator
\be
S_\lambda=\frac{1}{\lambda}\p_\lambda-\frac{1}{2\,\lambda^2}-\lambda,
\ee
and $c(\lambda)$ is the unique formal solution to the ODE problem
\bea
&&\left(S_\lambda^2-\lambda^2\right)c(\lambda)=0,\\
&&c(\lambda)=1+\mathcal{O}(\lambda^{-1}),\quad \lambda\rightarrow \infty.
\eea
The explicit expression for $c(\lambda)$ is given by
\be\label{def-c}
c(\lambda)=\sum_{k=0}^\infty c_k \,\lambda^{-3k}=1+\sum_{k=0}^\infty  \frac{(-1)^k}{288^k}\frac{(6k)!}{(3k)!(2k)!}\lambda^{-3k}.
\ee
Denote by $q(\lambda)$ the following formal series
\be q(\lambda)=-\frac{1}{\lambda}\,S_\lambda \,c(\lambda), \ee
which has the form
\be\label{def-q}
q(\lambda)=\sum_{k=0}^\infty q_k \,\lambda^{-3k}=1+\sum_{k=0}^\infty \frac{1+6k}{1-6k} \frac{(-1)^k}{288^k}\frac{(6k)!}{(3k)!(2k)!}\lambda^{-3k}.
\ee
It is easy to check that
\bea
&& W^{WK} = \Span_\mathbb{C}\left\{\lambda^{2k}c(\lambda), \lambda^{2k+1}q(\lambda)\right\}_{k=0}^{\infty},\\
&& \lambda^2\,W^{WK}\subset\,W^{WK}.
\eea
In principle, one can reconstruct from the data $\{c_k,\,q_k\}_{k\geq 0}$ the corresponding tau function $Z(\bt;\e)$; this was made explicit by Itzykson and Zuber in \cite{Itzykson_Zuber}.

Recently, by solving a set of linear constraints  (the Virasoro constraints) on $Z(\bt)$ in the Fermionic Fock space, Zhou \cite{Zhou} derived an alternative explicit formula for the Witten-Kontsevich tau function of the form
\begin{equation}
Z(\bt)=\sum_{\mu\in\mathbb{Y}} A_\mu^{Zhou} \cdot s_\mu(\bgT),
\end{equation}
where the summation is over all Young diagrams, $s_\mu(\bgT)$ is the Schur polynomial indexed by a partition $\mu$, and 
\be \bgT=(T_1,T_2,T_3,...)\ee are customarily called ``times" of the flows of the KP hierarchy, which satisfy
\be t_k=(-1)^k \cdot \sqrt{-2}\,T_{2k+1} \prod_{j=0}^k \left(j+\frac12\right).\ee
In the Frobenius notation $\mu=(m_1,...,m_k\ |\ n_1,...,n_k)$ \cite{Macdonald}, the coefficient $A_\mu^{Zhou}$ is defined by
\bea A_\mu^{Zhou}=(-1)^{n_1+...+n_k}\det_{1\leq i,j\leq k} (A_{m_i,\,n_j}^{Zhou})\eea
where $A_{i,j}^{Zhou}$ are complex numbers given by the following expressions
\begin{align}
A_{3m-1,3n}^{Zhou}=&\,A_{3m-3,3n+2}^{Zhou}=(-1)^n \left(-\frac{\sqrt{-2}}{144}\right)^{m+n}\frac{(6m+1)!!}{(2(m+n))!} \nn\\
&\cdot \prod_{j=0}^{n-1}(m+j)\cdot \prod_{j=1}^n(2m+2j-1)\cdot \left(B_n(m)+\frac{b_n}{6m+1}\right),\\
A_{3m-2,3n+1}^{Zhou}=&\,(-1)^{n+1} \left(-\frac{\sqrt{-2}}{144}\right)^{m+n}\frac{(6m+1)!!}{(2(m+n))!} \nn\\
&\cdot \prod_{j=0}^{n-1}(m+j)\cdot \prod_{j=1}^n(2m+2j-1)\cdot \left(B_n(m)+\frac{b_n}{6m-1}\right)
\end{align}
where $m\geq 1,n\geq 0,$ and $B_n(m)$ is a polynomial in $m$ of degree $n-1$ defined by
\be
B_n(x)=\frac{1}{6}\sum_{j=1}^n 108^j b_{n-j} \cdot (x+n)_{[j-1]},\quad x \in \C
\ee
with
\bea
&&(y)_{[j]}:=\frac{\Gamma(y+1)}{\Gamma(y-j+1)},\quad y\in\mathbb{C},\\
&&b_k:=\frac{2^k\cdot (6k+1)!!}{(2k)!},\quad k\in \mathbb{Z}_{\geq 0}.
\eea
\begin{theorem}\label{main}
The coefficients $A_{m,n}^{Zhou},\,m,n\geq 0$ are the {\em affine coordinates} of the point of the Sato Grassmannian corresponding to the Witten-Kontsevich tau function. Moreover, they satisfy the following properties:
\begin{itemize}
\item[(i)] Two-step recursion relation:
\be
A_{m+2,n}^{Zhou}-A_{m,n+2}^{Zhou}=A_{m,0}^{Zhou}A_{1,n}^{Zhou}+A_{m,1}^{Zhou}A_{0,n}^{Zhou},\quad \forall\,m,n\geq 0.
\ee
\item[(ii)] Symmetry:
\be
A_{n,m}^{Zhou}=(-1)^{m+n} A_{m,n}^{Zhou} ,\quad \forall\,m,n\geq 0.
\ee
\item[(iii)] Generating formula: Rescale\footnote{The rescaling is due to a different choice of the flow normalisations where our choice is more natural and has the advantage that all coefficients are rational; see \eqref{normalisation-choice}.}
\be B_{m,n}=(\sqrt{-2})^{m+n+1}A^{Zhou}_{m,n},\quad m,n\geq 0,\ee
and define $2\times 2$ matrices $Z_{m,n}$ by
\be\label{Z-KdV}
Z_{m,n}=\left(\begin{array}{cc}B_{2m+1,2n}& B_{2m+1,2n+1} \\ B_{2m,2m} & B_{2m,2n+1} \\ \end{array}\right);
\ee
Then we have
\be
\sum_{k,l=0}^\infty Z_{k,l}\,\alpha^{-k-1}\beta^{-l-1}=\frac{I-G(\alpha)G(\beta)^{-1}}{\alpha-\beta}
\ee
where $G(\lambda)$ is a matrix-valued formal series  defined by
\be\label{gkdv}
G(\lambda)=\left(\begin{array}{ll}
\sum_k c_{2k}\lambda^{-3k} &  \sum_k q_{2k+1} \lambda^{-3k-1}\\
\sum_k c_{2k+1}\lambda^{-3k-2} & \sum_k q_{2k} \lambda^{-3k}\\
\end{array}\right).
\ee
Here $\sum_k:=\sum_{k=0}^\infty$ and we recall that for $k\geq 0,$
\be c_k=\frac{(-1)^k}{288^k}\frac{(6k)!}{(3k)!(2k)!},\qquad q_k= \frac{1+6k}{1-6k}c_k. \ee
\end{itemize}
\end{theorem}

\begin{remark} The formulation of matrix-valued affine coordinates and tau functions holds true for generic Sato Grassmannian solutions to the KdV hierarchy. See Lemmae \ref{m-aff-1}, \ref{m-aff-2}, \ref{unique-Z} and Theorem \ref{aff-tau} for details.
\end{remark} 

\begin{theorem}\label{two-affine} The $R$-matrix of the $A_2$ Frobenius manifold can be expressed by
the $G$-matrix of the point of the Sato Grassmannian corresponding to the Witten-Kontsevich tau function in the following way:
\be
R(z) = z^{\frac{\sigma_3}{6}}G\left(z^{-\frac{2}{3}}\right)^{-1}z^{-\frac{\sigma_3}{6}},\qquad \sigma_3 =\begin{pmatrix}1&0\\0& -1\end{pmatrix},
\ee
where $G(\lambda)$ is defined in \eqref{gdef} which coincides with \eqref{gkdv} in the Witten-Kontsevich case, and $R(z)$ is the $R$-matrix in the notation of \cite{PPZ} evaluated at $\phi = 6^{-2/3}.$ Furthermore, the $V$-matrices associated to $R(z)$ can be expressed by matrix-valued affine coordinates: For any $k,l\geq 0,$ 
\begin{align}
&V_{2k,2l+1}=Z_{3k,\,3l+2},\,\quad\quad V_{2l+1,2k}=-Z_{3l+2,3k},\label{a1}\\
&V_{2k,2l}=-Z_{3k,3l+1}-Z_{3k,3l},\,\quad V_{2k+1,2l+1}=Z_{3k+2,3l+2}+Z_{3k+2,3l+1}\label{a2}
\end{align}
where $V_{i,j}$ are defined in \eqref{V-def}, and $Z_{i,j}$ are defined in \eqref{Z-def} which coincide with \eqref{Z-KdV} in the Witten-Kontsevich case.
\end{theorem}

The paper is organized as follows: In Sect.\,\ref{sec:sato_matrix_aff_coords}, we recall the Sato Grassmannian for the KdV hierarchy and introduce matrix-valued affine coordinates, followed by Sato's definition of the tau function in Sect.\,\ref{sec:tau}. In Sect.\,\ref{sec:WK}, we apply the general construction worked out in Sect.~\ref{sec:sato_matrix_aff_coords} and \ref{sec:tau} to the particular case of the Witten-Kontsevich tau function and prove the above theorems. In Sect.\,\ref{sec:A2_Frobenius}, we give the precise description of Thm.\,\ref{two-affine}. Further remarks are given in Sect.\,\ref{sec:conclusion}.

\section{Sato Grassmannian for the KdV hierarchy and matrix-valued affine coordinates}
\label{sec:sato_matrix_aff_coords}
Let $\C[[\lambda^{-1}]]$ denote the linear space of formal series with finitely many terms of positive powers, and consider the Sato Grassmannian $GM$ as defined originally in \cite{Sato}. 

A point $W\in GM$ is a subspace of $\C[[\lambda^{-1}]]$ which can be written as a linear span of a set of basis vectors. A particularly useful choice of basis of $W$ is of the following form 
\be\label{def-affine}
W = \Span_\mathbb{C}\left\{\lambda^l +\sum_{k=0}^{\infty}A_{k,l} \lambda^{-k-1}\right\}_{l=0}^{\infty},
\ee
where the coefficients $A_{k,l}$ are called the \textit{affine coordinates} of $W$. Such a choice is not always possible: existence of basis of form \eqref{def-affine} characterizes the \emph{big cell} of $GM$ (see \cite{Segal_Wilson} for the details). Given a point $W$ in $GM,$ it should be noted that the affine coordinates 
\be
\{A_{k,l}\, | \,k,l\geq 0\}
\ee 
(if exist) must be unique. 

Let us consider the infinite dimensional submanifold 
\be
GM_2:= \{W\in GM\ \colon\ \lambda^2W \subset W\} \subset GM
\ee
associated to the KdV hierarchy. Generic points of $GM_2$ can be paramertrized by two formal series in $\C[[\lambda^{-1}]]$ with leading terms $1$. To be precise, for any point $W\in GM_2$ satisfying the big cell condition, there exists
\be
a(\lambda) = 1+\sum_{k=1}^{\infty}a_k\, \lambda^{-k}, \qquad b(\lambda) = 1+\sum_{k=1}^{\infty}b_k\, \lambda^{-k}
\ee
such that
\be
W = \Span_\mathbb{C}\left\{\lambda^{2k}a(\lambda), \lambda^{2k+1}b(\lambda)\right\}_{k=0}^{\infty}.
\ee
As above denote by $A_{i,j}$ the corresponding affine coordinates. In the infinite matrix notation, the subspace $W$ can be visualized as an infinite collection of doubly infinite column vectors (see Figure \ref{fig:aff_coords}).

\begin{figure}[htb!]
\begin{center}
\begin{tikzpicture}[scale=0.65]
\draw[very thick] (0,0)--(4,0);
\draw (0,2) rectangle (2,4);
\draw (0,0) rectangle (2,2);
\draw (0,-2) rectangle (2,0);
\node at (0.5,5) {$\vdots$};
\node at (1.5,5) {$\vdots$};
\node at (2.5,5) {$\vdots$};
\node at (3.5,5) {$\vdots$};
\node at (5,5) {$\iddots$};
\node at (5,3.5) {$\cdots$};
\node at (5,2.5) {$\cdots$};
\node at (5,1.5) {$\cdots$};
\node at (5,0.5) {$\cdots$};
\node at (5,-0.5) {$\cdots$};
\node at (5,-1.5) {$\cdots$};
\node at (5,-2.5) {$\cdots$};
\node at (5,-3.5) {$\cdots$};
\node at (5,-4.5) {$\ddots$};
\node at (0.5,2.5) {$a_3$};
\node at (1.5,2.5) {$b_4$};
\node at (0.5,3.5) {$a_4$};
\node at (1.5,3.5) {$b_5$};
\node at (0.5,0.5) {$a_1$};
\node at (1.5,0.5) {$b_2$};
\node at (0.5,1.5) {$a_2$};
\node at (1.5,1.5) {$b_3$};
\node at (0.5,-0.5) {$1$};
\node at (1.5,-0.5) {$b_1$};
\node at (0.5,-1.5) {$0$};
\node at (1.5,-1.5) {$1$};
\draw (2,2) rectangle (4,4);
\draw (2,0) rectangle (4,2);
\draw (2,-2) rectangle (4,0);
\draw (2,-4) rectangle (4,-2);
\node at (2.5,2.5) {$a_5$};
\node at (3.5,2.5) {$b_6$};
\node at (2.5,3.5) {$a_6$};
\node at (3.5,3.5) {$b_7$};
\node at (2.5,0.5) {$a_3$};
\node at (3.5,0.5) {$b_4$};
\node at (2.5,1.5) {$a_4$};
\node at (3.5,1.5) {$b_5$};
\node at (2.5,-1.5) {$a_1$};
\node at (3.5,-1.5) {$b_2$};
\node at (2.5,-0.5) {$a_2$};
\node at (3.5,-0.5) {$b_3$};
\node at (2.5,-2.5) {$1$};
\node at (3.5,-2.5) {$b_1$};
\node at (2.5,-3.5) {$0$};
\node at (3.5,-3.5) {$1$};
\end{tikzpicture}
\begin{tikzpicture}[scale=0.65]
\draw[color=white] (-1,5)--(1,-5);
\node at (0,0) {$\leadsto$};
\end{tikzpicture}
\begin{tikzpicture}[scale=0.65]
\draw[very thick] (0,0)--(4,0);
\draw (0,2) rectangle (2,4);
\draw (0,0) rectangle (2,2);
\draw (0,-2) rectangle (2,0);
\node at (0.5,5) {$\vdots$};
\node at (1.5,5) {$\vdots$};
\node at (2.5,5) {$\vdots$};
\node at (3.5,5) {$\vdots$};
\node at (5,5) {$\iddots$};
\node at (5,3.5) {$\cdots$};
\node at (5,2.5) {$\cdots$};
\node at (5,1.5) {$\cdots$};
\node at (5,0.5) {$\cdots$};
\node at (5,-0.5) {$\cdots$};
\node at (5,-1.5) {$\cdots$};
\node at (5,-2.5) {$\cdots$};
\node at (5,-3.5) {$\cdots$};
\node at (5,-4.5) {$\ddots$};
\node at (0.5,2.5) {$A_{2,0}$};
\node at (1.5,2.5) {$A_{2,1}$};
\node at (0.5,3.5) {$A_{3,0}$};
\node at (1.5,3.5) {$A_{3,1}$};
\node at (0.5,0.5) {$A_{0,0}$};
\node at (1.5,0.5) {$A_{0,1}$};
\node at (0.5,1.5) {$A_{1,0}$};
\node at (1.5,1.5) {$A_{1,1}$};
\node at (0.5,-0.5) {$1$};
\node at (1.5,-0.5) {$0$};
\node at (0.5,-1.5) {$0$};
\node at (1.5,-1.5) {$1$};
\draw (2,2) rectangle (4,4);
\draw (2,0) rectangle (4,2);
\draw (2,-2) rectangle (4,0);
\draw (2,-4) rectangle (4,-2);
\node at (2.5,2.5) {$A_{2,2}$};
\node at (3.5,2.5) {$A_{2,3}$};
\node at (2.5,3.5) {$A_{3,2}$};
\node at (3.5,3.5) {$A_{3,3}$};
\node at (2.5,0.5) {$A_{0,2}$};
\node at (3.5,0.5) {$A_{0,3}$};
\node at (2.5,1.5) {$A_{1,2}$};
\node at (3.5,1.5) {$A_{1,3}$};
\node at (2.5,-1.5) {$0$};
\node at (3.5,-1.5) {$0$};
\node at (2.5,-0.5) {$0$};
\node at (3.5,-0.5) {$0$};
\node at (2.5,-2.5) {$1$};
\node at (3.5,-2.5) {$0$};
\node at (2.5,-3.5) {$0$};
\node at (3.5,-3.5) {$1$};
\end{tikzpicture}
\end{center}
\caption{Affine coordinates on the big cell of $GM_2$}
\label{fig:aff_coords}
\end{figure}

Let us now define the formal loop group element associated to the subspace $W$ \`a la Segal--Wilson \cite{Segal_Wilson} as
\be
G(\lambda) =\sum_{k=0}^{\infty}G_k\, \lambda^{-k} \in \mathfrak{gl}(2,\C)[[\lambda^{-1}]]
\ee
by
\begin{align}
G(\lambda)_{11}& =\sum_{k=0}^\infty a_{2k}\, \lambda^{-k}, \quad G(\lambda)_{12} =\sum_{k=0}^\infty b_{2k+1}\,\lambda^{-k},\\
G(\lambda)_{21}& =\sum_{k=1}^\infty a_{2k-1}\, \lambda^{-k},\quad  G(\lambda)_{22} =\sum_{k=0}^\infty b_{2k}\, \lambda^{-k},
\end{align}
(see Figure \ref{fig:loop_rep}), or equivalently by
\be
G(\lambda) = \frac{1}{2}\,\lambda^{\frac{1}{4}\sigma_3}
\begin{pmatrix}
1 & 1\\
1 & -1
\end{pmatrix}
\begin{pmatrix}\label{gdef}
a(\sqrt{\lambda}) & b(\sqrt{\lambda})\\
a(-\sqrt{\lambda}) & -b(-\sqrt{\lambda})
\end{pmatrix}
\lambda^{-\frac{1}{4}\sigma_3}
\ee
where
\be
\sigma_3=\begin{pmatrix}
1 & 0\\
0 & -1
\end{pmatrix}.
\ee
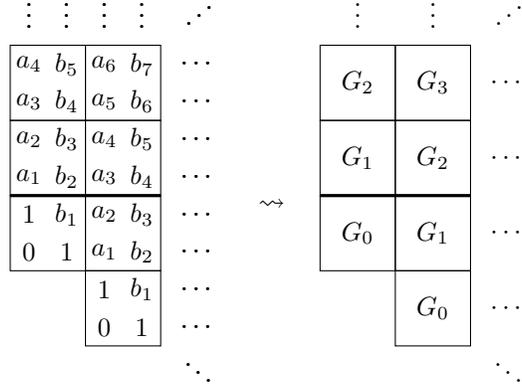
\begin{figure}[htb!]
\begin{center}
\begin{tikzpicture}[scale=0.5]
\draw[very thick] (0,0)--(4,0);
\draw (0,2) rectangle (2,4);
\draw (0,0) rectangle (2,2);
\draw (0,-2) rectangle (2,0);
\node at (0.5,5) {$\vdots$};
\node at (1.5,5) {$\vdots$};
\node at (2.5,5) {$\vdots$};
\node at (3.5,5) {$\vdots$};
\node at (5,5) {$\iddots$};
\node at (5,3.5) {$\cdots$};
\node at (5,2.5) {$\cdots$};
\node at (5,1.5) {$\cdots$};
\node at (5,0.5) {$\cdots$};
\node at (5,-0.5) {$\cdots$};
\node at (5,-1.5) {$\cdots$};
\node at (5,-2.5) {$\cdots$};
\node at (5,-3.5) {$\cdots$};
\node at (5,-4.5) {$\ddots$};
\node at (0.5,2.5) {$a_3$};
\node at (1.5,2.5) {$b_4$};
\node at (0.5,3.5) {$a_4$};
\node at (1.5,3.5) {$b_5$};
\node at (0.5,0.5) {$a_1$};
\node at (1.5,0.5) {$b_2$};
\node at (0.5,1.5) {$a_2$};
\node at (1.5,1.5) {$b_3$};
\node at (0.5,-0.5) {$1$};
\node at (1.5,-0.5) {$b_1$};
\node at (0.5,-1.5) {$0$};
\node at (1.5,-1.5) {$1$};
\draw (2,2) rectangle (4,4);
\draw (2,0) rectangle (4,2);
\draw (2,-2) rectangle (4,0);
\draw (2,-4) rectangle (4,-2);
\node at (0.5,5) {$\vdots$};
\node at (1.5,5) {$\vdots$};
\node at (2.5,5) {$\vdots$};
\node at (3.5,5) {$\vdots$};
\node at (5,5) {$\iddots$};
\node at (5,3.5) {$\cdots$};
\node at (5,2.5) {$\cdots$};
\node at (5,1.5) {$\cdots$};
\node at (5,0.5) {$\cdots$};
\node at (5,-0.5) {$\cdots$};
\node at (5,-1.5) {$\cdots$};
\node at (5,-2.5) {$\cdots$};
\node at (5,-3.5) {$\cdots$};
\node at (5,-4.5) {$\ddots$};
\node at (2.5,2.5) {$a_5$};
\node at (3.5,2.5) {$b_6$};
\node at (2.5,3.5) {$a_6$};
\node at (3.5,3.5) {$b_7$};
\node at (2.5,0.5) {$a_3$};
\node at (3.5,0.5) {$b_4$};
\node at (2.5,1.5) {$a_4$};
\node at (3.5,1.5) {$b_5$};
\node at (2.5,-1.5) {$a_1$};
\node at (3.5,-1.5) {$b_2$};
\node at (2.5,-0.5) {$a_2$};
\node at (3.5,-0.5) {$b_3$};
\node at (2.5,-2.5) {$1$};
\node at (3.5,-2.5) {$b_1$};
\node at (2.5,-3.5) {$0$};
\node at (3.5,-3.5) {$1$};
\end{tikzpicture}
\begin{tikzpicture}[scale=0.5]
\draw[color=white] (-1,5)--(1,-5);
\node at (0,0) {$\leadsto$};
\end{tikzpicture}
\begin{tikzpicture}[scale=0.5]
\draw[very thick] (0,0)--(4,0);
\draw (0,2) rectangle (2,4);
\draw (0,0) rectangle (2,2);
\draw (0,-2) rectangle (2,0);
\node at (1,5) {$\vdots$};
\node at (3,5) {$\vdots$};
\node at (5,5) {$\iddots$};
\node at (5,3) {$\cdots$};
\node at (5,1) {$\cdots$};
\node at (5,-1) {$\cdots$};
\node at (5,-3) {$\cdots$};
\node at (5,-4.5) {$\ddots$};
\node at (1,3) {$G_2$};
\node at (1,1) {$G_1$};
\node at (1,-1) {$G_0$};
\draw (2,2) rectangle (4,4);
\draw (2,0) rectangle (4,2);
\draw (2,-2) rectangle (4,0);
\draw (2,-4) rectangle (4,-2);
\node at (3,3) {$G_3$};
\node at (3,1) {$G_2$};
\node at (3,-1) {$G_1$};
\node at (3,-3) {$G_0$};
\end{tikzpicture}
\end{center}
\caption{Condensed notation of the subspace data in terms of $2\times 2$ matrices}
\label{fig:loop_rep}
\end{figure}

Below without loss of generality we assume that
\be
G_0 =I,
\ee
that is, $b_1=0$. Indeed,  this can always be achieved by a right multiplication by an upper unitriangular constant matrix that does not change $W$.

\begin{definition}\label{Z-def}
The $2\times 2$ matrices 
\be
Z_{k,l} :=
\begin{pmatrix}
A_{2k+1,2l} & A_{2k+1,2l+1}\\
A_{2k,2l} & A_{2k,2l+1}
\end{pmatrix},\qquad k,l\geq 0,
\ee
will be referred to as the \emph{matrix-valued affine coordinates} of $W \in GM_2$
where $A_{m,n}$ are the affine coordinates of $W \in GM$.
\end{definition}
\begin{figure}[htb!]
\begin{center}
\begin{tikzpicture}[scale=0.5]
\draw[very thick] (0,0)--(6,0);
\draw (0,2) rectangle (2,4);
\draw (0,0) rectangle (2,2);
\draw (0,-2) rectangle (2,0);
\node at (1,5) {$\vdots$};
\node at (3,5) {$\vdots$};
\node at (5,5) {$\vdots$};
\node at (7,5) {$\iddots$};
\node at (7,3) {$\cdots$};
\node at (7,1) {$\cdots$};
\node at (7,-1) {$\cdots$};
\node at (7,-3) {$\cdots$};
\node at (7,-5) {$\cdots$};
\node at (7,-6.5) {$\ddots$};
\node at (1,3) {$G_2$};
\node at (1,1) {$G_1$};
\node at (1,-1) {$G_0$};
\draw (2,2) rectangle (4,4);
\draw (2,0) rectangle (4,2);
\draw (2,-2) rectangle (4,0);
\draw (2,-4) rectangle (4,-2);
\node at (3,3) {$G_3$};
\node at (3,1) {$G_2$};
\node at (3,-1) {$G_1$};
\node at (3,-3) {$G_0$};
\draw (4,2) rectangle (6,4);
\draw (4,0) rectangle (6,2);
\draw (4,-2) rectangle (6,0);
\draw (4,-4) rectangle (6,-2);
\draw (4,-6) rectangle (6,-4);
\node at (5,3) {$G_4$};
\node at (5,1) {$G_3$};
\node at (5,-1) {$G_2$};
\node at (5,-3) {$G_1$};
\node at (5,-5) {$G_0$};
\end{tikzpicture}
\begin{tikzpicture}[scale=0.5]
\draw[color=white] (-1,5)--(1,-7);
\node at (0,0) {$\leadsto$};
\end{tikzpicture}
\begin{tikzpicture}[scale=0.5]
\draw[very thick] (0,0)--(6,0);
\draw (0,2) rectangle (2,4);
\draw (0,0) rectangle (2,2);
\draw (0,-2) rectangle (2,0);
\node at (1,5) {$\vdots$};
\node at (3,5) {$\vdots$};
\node at (5,5) {$\vdots$};
\node at (7,5) {$\iddots$};
\node at (7,3) {$\cdots$};
\node at (7,1) {$\cdots$};
\node at (7,-1) {$\cdots$};
\node at (7,-3) {$\cdots$};
\node at (7,-5) {$\cdots$};
\node at (7,-6.5) {$\ddots$};
\node at (1,3) {$Z_{10}$};
\node at (1,1) {$Z_{00}$};
\node at (1,-1) {$I$};
\draw (2,2) rectangle (4,4);
\draw (2,0) rectangle (4,2);
\draw (2,-2) rectangle (4,0);
\draw (2,-4) rectangle (4,-2);
\node at (3,3) {$Z_{11}$};
\node at (3,1) {$Z_{01}$};
\node at (3,-1) {$0$};
\node at (3,-3) {$I$};
\draw (4,2) rectangle (6,4);
\draw (4,0) rectangle (6,2);
\draw (4,-2) rectangle (6,0);
\draw (4,-4) rectangle (6,-2);
\draw (4,-6) rectangle (6,-4);
\node at (5,3) {$Z_{12}$};
\node at (5,1) {$Z_{02}$};
\node at (5,-1) {$0$};
\node at (5,-3) {$0$};
\node at (5,-5) {$I$};
\end{tikzpicture}
\end{center}
\label{fig:matrix-affine}
\caption{Matrix-valued affine coordinates}
\end{figure}
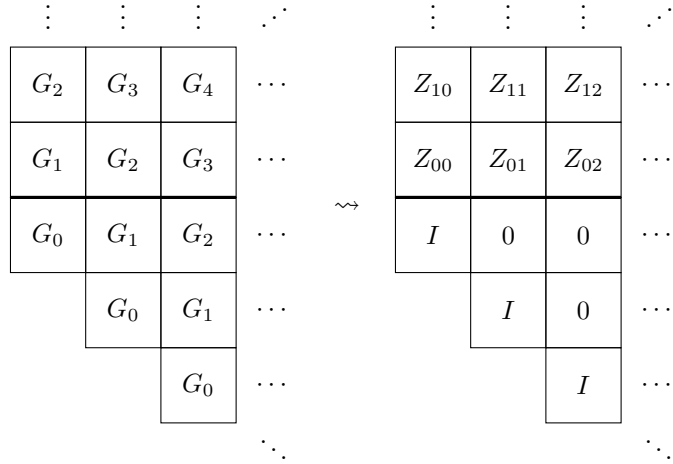

The inverse of $G(\lambda)$ is
\be
G(\lambda)^{-1}= \frac{1}{2\det(G(\lambda))}\,\lambda^{1/4\sigma_3}
\begin{pmatrix}
b(-\sqrt{\lambda}) & b(\sqrt{\lambda})\\
a(-\sqrt{\lambda}) & -a(\sqrt{\lambda})
\end{pmatrix}
\begin{pmatrix}
1 & 1\\
1 & -1
\end{pmatrix}
\lambda^{-1/4\sigma_3 },
\ee
whose formal series expansion is of the form
\be
G(\lambda)^{-1} = I +\sum_{k=0}^{\infty}U_k\, \lambda^{-k}.
\ee
\begin{lemma} The matrix-valued affine coordinates $Z_{k,l}$ can be read off from the series
\be
\label{eq:Z_generating}
G(\lambda)\left(\lambda^k G(\lambda)^{-1}\right)_{+} =\lambda^k+ \sum_{l=0}^{\infty}Z_{l,k}\lambda^{-l-1} ,
\ee
where $(\,)_{+}$ stands for taking the polynomial part of a formal series in ${\mathfrak{gl}(2,\C)}[[\lambda^{-1}]]$.
\end{lemma}
\begin{proof} Since
\begin{align}
G(\lambda)\left(\lambda^k G(\lambda)^{-1}\right)_{+} &= G(\lambda)(\lambda^k+U_1 \lambda^{k-1}+\dots+U_k)\\
&= \lambda^{k}G(\lambda)\left(G(\lambda)^{-1}+\order{\lambda^{-{k+1}}}\right)\\
&= \lambda^{k}+\order{\lambda^{-1}},
\end{align} 
the formal Laurent series on the l.h.s.~of \eqref{eq:Z_generating} is indeed of the required form.
The rest follows from Gauss elimination (which is equivalent to multiplication on the right by the block matrix $[U_{j-i}]_{i,j=1}^{\infty}$).
\end{proof}

\begin{lemma}\label{m-aff-1} The matrix-valued affine coordinates $Z_{k,l}$ have the following expressions
\be
\label{eq:Z}
Z_{k,l} = -\sum_{j=0}^{k}G_{j}U_{k+l+1-j},\qquad \forall\,k,l\geq 0.
\ee
Moreover, the following recursion formula holds true:
\be
\label{eq:Z_recursion}
Z_{k+1,l}-Z_{k,l+1} = Z_{k,0}Z_{0,l}, \qquad \forall\,k,l\geq 0.
\ee
\end{lemma}
\begin{proof} The formula \eqref{eq:Z} follows from the Gauss elimination procedure. It is straightforward to prove the recursion relation \eqref{eq:Z}:
\begin{align}
\nonumber
Z_{k+1,l}-Z_{k,l+1} &= -\sum_{j=0}^{k+1}G_{j}U_{k+l+2-j}+\sum_{j=0}^{k}G_{j}U_{k+l+2-j}\\
&= -G_{k+1}U_{l+1}= Z_{k,0}Z_{0,l}.
\end{align}
\end{proof}

It is clear that we have
\begin{lemma}\label{unique-Z}
The matrix-valued affine coordinates $Z_{k,l}$ can be uniquely specified by the recursion relations \eqref{eq:Z_recursion} together with knowledge of the boundary data $Z_{k,0}$ and $Z_{0,l},$ where 
\be\label{ini-Z}
Z_{k,0}= G_{k+1},\quad Z_{0,l}=-U_{l+1}.
\ee
\end{lemma}
Note that the initial condition for $Z_{0,0}$ is consistent since $U_{1} = -G_{1}$.
\begin{lemma}\label{m-aff-2}
The following formula holds true for the matrix-valued affine coordinates:
\be\label{G-Z}
 \frac{I-G(\alpha)G(\beta)^{-1}}{\alpha-\beta}=\sum_{k,l=0}^{\infty}Z_{k,l}\,\alpha^{-k-1}\beta^{-l-1}.
\ee
\end{lemma}
\begin{proof} Follows from the recursion \eqref{eq:Z_recursion}. \end{proof}

\noindent{\bf Symmetry.} In the case that $\det G(\lambda) \equiv 1,$ we have
\bea
&&G(\lambda)^{-1} = \sigma_2 \cdot G(\lambda)^{T}\cdot \sigma_{2},\\
&&U_{k}= \sigma_2 \cdot G_k^{T} \cdot \sigma_2,\qquad \forall\,k\geq 0.
\eea
where \be \sigma_2=\begin{pmatrix}
0 & -\sqrt{-1}\\
\sqrt{-1} & 0
\end{pmatrix}.
\ee
\begin{lemma}\label{m-symm}
If $\det{G(\lambda)}\equiv1,$ the matrix-valued affine coordinates $Z_{k,l}$ have the following expressions
\be
Z_{k,l} = -\sum_{j=0}^{k}G_{j}\cdot\sigma_2 \cdot (G_{k+l+1-j})^{T}\cdot \sigma_2;
\ee
Moreover, they possess the symmetry
\be
\label{eq:Z_symmetry}
Z_{l,k} = -\sigma_2 \cdot Z_{k,l}^{T} \cdot \sigma_2.
\ee
\end{lemma}
\begin{proof}
The symmetry \eqref{eq:Z_symmetry} is justified by
\begin{align}
Z_{l,k}+\sigma_2 Z_{k,l}^{T}\sigma_2^{-1}&=-\left(\sum_{j=0}^{l}G_{j}U_{k+l+1-j} +\sum_{i=0}^{k}\sigma_2U_{l+k+1-i}^{T}\sigma_2^{-1}\sigma_2G_{i}^{T}\sigma_2^{-1}\right)\nn\\
&=-\left(\sum_{j=0}^{l}G_{j}U_{k+l+1-j} +\sum_{i=0}^{k}G_{l+k+1-k}U_{i}\right) =0.
\end{align}
\end{proof}

\section{Sato's tau function in terms of affine coordinates}
\label{sec:tau}
In this section we will give the definition of the tau function corresponding to an arbitrary point $W$ of the Sato Grassmannian for the KdV hierarchy.

Let us recall the original approach of Kontsevich \cite{Kontsevich}, which was explained and clarified in  \cite{Itzykson_Zuber}. As in the previous section, let $a(\lambda),b(\lambda)$ be the formal power series such that
\be
W = \Span_\mathbb{C}\left\{\lambda^{2k}a(\lambda), \lambda^{2k+1}b(\lambda)\right\}_{k=0}^{\infty}.
\ee
Denote
\begin{align}
f_{2k}(\lambda) &= \lambda^{2k}a(\lambda),\\
f_{2k+1}(\lambda) &= \lambda^{2k+1}b(\lambda),
\end{align}
and define, for $N\geq 1$,
\be
\label{eq:tau_N_initial}
\tilde\tau_{W,N}\left(\bx\right):= \frac{\displaystyle\det_{1\leq i,j\leq N}(f_j(x_i^{-1}))}{\displaystyle\det_{1\leq i,j\leq N}(x_i^{1-j})}, \qquad \bx=(x_1,...,x_N).
\ee

Recall that a partition $\mu=\left(\mu_1,\mu_2,...\right)$ is a sequence of weakly decreasing non-negative integers with $\mu_k=0$
 for sufficiently large $k$.
 The length $\ell(\mu)$ is the number of the non-zero parts of $\mu$ and the weight $|\mu|:=\mu_1+\mu_2+...$. The
Schur polynomial $s_{\mu}(\btheta)$ associated to $\mu$ is a polynomial in the variables
\be
\btheta: = (\theta_1,\theta_2,\dots),
\ee
defined as
\be
s_{\mu}(\btheta) = \det_{1\leq i,j \leq \ell(\lambda)}(h_{\lambda_i-i+j}(\btheta)),
\ee
where the polynomials $h_k(\btheta)$ are defined by the generating function
\be
\sum_{k=0}^{\infty}h_k(\btheta)z^n =e^{\sum_{j=1}^{\infty}\theta_j z^j}. 
\ee
If, for some $\bx=(x_1,...,x_N)$, $\btheta$ is of the special form
\be
\theta_k(\bx) = \frac{1}{k}\sum_{j=1}^{N}x_j^{k},\quad k\geq 1,
\ee
then the components of $\btheta=\btheta(\bx)$ are called \emph{Miwa variables}. The Schur polynomial $s_{\mu}$ is expressible in terms of $\bx$ as
\be
s_{\mu}(\btheta(\bx)) = \frac{\displaystyle\det_{1\leq i,j\leq N}(x_i^{l_j})}{\displaystyle\det_{1\leq i,j\leq N}(x_i^{N-j})},
\ee
where
\be
l_j: = \mu_j-j+N, \qquad 1 \leq j\leq N.
\ee
By the formal series version of the Cauchy--Binet identity \cite{Gantmacher}, we have
\begin{align}
\label{eq:Schur_exp_non_sym_1}
\tilde\tau_{W,N}\left({\mathbf x}\right) &= \sum_{0\leq l_{N} < l_{N-1}<\cdots < l_{1}}
\displaystyle\det_{1\leq i,j\leq N}\left(F^{(N)}_{l_i,j}\right)\frac{\displaystyle\det_{1\leq i,j\leq N}(x_i^{l_j})}{\displaystyle\det_{1\leq i,j\leq N}(x_i^{N-j})}\\
\label{eq:Schur_exp_non_sym_2}
&=\sum_{\substack{\mu\in {\mathbb Y}\\\ell(\mu)\leq N}}\displaystyle\det_{1\leq i,j\leq N}\left(F^{(N)}_{\mu_i-i+N,j}\right)s_{\mu}\left(\btheta(\bx)\right),
\end{align}
where $F^{(N)}$ is the $\infty \times N$ matrix
\be
\begin{tikzpicture}[scale=0.5]
\node at (-0.5,0) {$F^{(N)}=$};
\draw (1,-4) rectangle (5,0);
\draw (1,0) rectangle (5,5);
\node at (1.5,-0.5) {$a_0$};
\node at (1.5,0.5) {$a_1$};
\node at (1.5,1.5) {$a_2$};
\node at (1.5,3) {$\vdots$};
\node at (2.5,-1.5) {$b_0$};
\node at (2.5,-0.5) {$b_1$};
\node at (2.5,0.5) {$b_2$};
\node at (2.5,1.5) {$b_3$};
\node at (2.5,3) {$\vdots$};
\node at (3.5,-2.5) {$a_0$};
\node at (3.5,-1.5) {$a_1$};
\node at (3.5,-0.5) {$a_2$};
\node at (3.5,0.5) {$a_3$};
\node at (3.5,1.5) {$a_4$};
\node at (3.5,3) {$\vdots$};
\node at (4.5,3) {$\iddots$};
\node at (4.5,1.5) {$\dots$};
\node at (4.5,0.5) {$\dots$};
\node at (4.5,-0.5) {$\dots$};
\node at (4.5,-1.5) {$\dots$};
\node at (4.5,-2.5) {$\dots$};
\node at (4.5,-3.5) {$\ddots$};
\draw[snake=brace,below=1pt] (5,-4) -- node[below]{$N$} (1,-4);
\draw[snake=brace,right=1pt] (5,0) -- node[right]{$N$} (5,-4);
\node at (6,-0.5) {.};
\end{tikzpicture}
\ee
Eqs.~\eqref{eq:Schur_exp_non_sym_1} and \eqref{eq:Schur_exp_non_sym_2} were obtained in \cite{Itzykson_Zuber} and similar formulae also appeared in \cite{Cafasso} by using block Toeplitz determinants.

Note that
\be
\det_{1\leq i,j\leq N}\left(F^{(N)}_{\mu_i-i+N,j}\right)
\ee
is the \emph{Pl\"ucker coordinate} of the (infinite) Grassmannian frame matrix $F^{(N)}$ corresponding to a partition $\mu$ with $\ell(\mu)\leq N$. Therefore, according to Sato \cite{Sato}, the function
\be
\label{eq:tau_N_general}
\tau_{W,N}(\btheta):=\sum_{\substack{\mu\in {\mathbb Y}\\\ell(\mu)\leq N}}\displaystyle\det_{1\leq i,j\leq N}\left(F^{(N)}_{\mu_i-i+N,j}\right)s_{\mu}\left(\btheta\right)
\ee
is a tau function of the KP hierarchy with the independent variables $\btheta$. Also, it is easy to verify that $\tilde\tau_{W,N}(\bx) = \tau_{W,N}(\btheta(\bx))$.

Introduce a gradation for the formal power series ring $\mathbb{C}[[\btheta]]$ by using the degree assignments 
\be\deg \theta_k:=k,\quad k\geq 1.\ee
Then we have
\be
\tau_{W,N}(\btheta) = \sum_{k=0}^{\infty}\tau_{W,N}^{(k)}(\btheta)
\ee
where $\tau_{W,N}^{(0)}(\btheta)=1$ and $\tau_{W,N}^{(k)}(\btheta)$ for $k\geq 1$ is a graded homogeneous polynomial of $\btheta$ of degree $k$ which admits the form
\be
\tau_{W,N}^{(k)}(\btheta)=\sum_{\substack{\mu\in {\mathbb Y}\\\ell(\mu)\leq N,\ |\mu|=k}}\displaystyle\det_{1\leq i,j\leq N}\left(F^{(N)}_{\mu_i-i+N,j}\right)s_{\mu}\left(\btheta\right).
\ee
It is easy to see that $\tau_{W,N}^{(k)}(\btheta)$ depends only on the variables $\theta_1,...,\theta_k$ and that
\be
\tau_{W,N}^{(k)}(\btheta)=\tau_{W,k}^{(k)}(\btheta) \qquad \text{for } N\geq k.
\ee
Similarly as in \cite{Itzykson_Zuber}, this allows to define without ambiguity the formal power series
\be
\tau_{W}(\btheta) = \sum_{k=0}^{\infty}\tau_{W,k}^{(k)}(\btheta).
\ee
This formal series $\tau_{W}(\btheta)$ is Sato's tau function of the KP hierarchy corresponding to the subspace $W$. Moreover, the constraint $\lambda^2W\subset W$ implies that the flows corresponding to the even variables $\theta_{2j}$ are trivial, i.e., $\tau_{W}(\btheta)$ is a tau function of the KdV hierarchy.

By employing elementary column operations on $F^{(N)}$ in equation \eqref{eq:tau_N_general} we obtain
\be\label{eq:Giambelli}
\det_{1\leq i,j\leq N}\left(F^{(N)}_{\mu_i-i+N,j}\right) = (-1)^{n_1+\dots+n_k}\det_{1\leq i,j\leq k}\left(A_{m_i,n_j}\right):=A_{\mu} 
\ee
where the partition $\mu$ is expressed in terms of its Frobenius characteristics \cite{Macdonald} $\mu=(m_1,\dots, m_k\, |\,n_1,\dots,n_k)$.
Equation \eqref{eq:Giambelli} will be referred to as a \emph{Giambelli-type formula}\footnote{The standard Giambelli formula \cite{Macdonald} says $$s_{(m_1,\dots, m_k |n_1,\dots,n_k)}(\btheta) = \det_{1\leq i,j\leq k}(s_{(m_i|n_j)}(\btheta)).$$}, as in \cite{HE}. We arrive at

\begin{proposition} \label{aff-tau}
Given an arbitrary point $W$ of the Sato Grassmannian such that $\lambda^2 W\subset W$, the formal series defined by
\be \label{tau-kdv}  \tau_{W}(\btheta)=\sum_{\mu\in\mathbb{Y}} A_\mu s_\mu({\boldsymbol\theta}) \ee
is the tau function of the KdV hierarchy \eqref{KdV-hie1},\eqref{KdV-hie2} corresponding to $W$. Here $A_\mu$ is defined in \eqref{eq:Giambelli} with the affine coordinates $A_{m_i,n_j}$ defined in \eqref{def-affine}, $s_\mu(\btheta)$ is the Schur polynomial in $\btheta$ associated to a partition $\mu,$ and 
\be\label{normalisation-choice}
t_k = -(2k+1)!!\,\theta_{2k+1},\qquad k\geq 0.
\ee
\end{proposition}

\begin{example} As a simple example, consider
\be
G(\lambda) = \begin{pmatrix}1& c\,\lambda^{-1}\\ 0& 1\end{pmatrix}.
\ee
This gives the affine coordinates
\be
A_{m,n} = c\,\delta_{m,1}\delta_{n,1},
\ee
and therefore
\be
\tau_{W}(\btheta)= 1 - c\,s_{(2,1)}(\btheta)=1- c\left(\frac{\theta_1^3}{3}-\theta_3\right)=1+c\frac{t_0^3}{3}-c\frac{ t_1}{3}.
\ee
The corresponding solution to the KdV hierarchy \eqref{KdV-hie1},\eqref{KdV-hie2} is given by
\be
u(\bt) = \,\partial^2_{t_0}\log\tau_{W}(\btheta)=-\frac{3\, c\, t_0 \left(c \left(t_0^3+2\, t_1\right)-6\right)}{\left(c \left(t_0^3-t_1\right)+3\right)^2}.
\ee
It satisfies the following initial condition
\be
u|_{t_{\geq 1}=0}=-\frac{3\, c\, t_0 \left(c\,t_0^3-6\right)}{\left(c\,t_0^3+3\right)^2}=2 c\, t_0-\frac{5 c^2 t_0^4}{3}+\frac{8 c^3 t_0^7}{9}-\frac{11}{27} c^4 t_0^{10}+o(t_0^{10}),\quad t_0\rightarrow 0.
\ee
We point out the following interesting property of this example: The function $$-\frac{3\, c\, t_0 \left(c\,t_0^3-6\right)}{\left(c\,t_0^3+3\right)^2}$$ is a common solution to all the higher order ($\text{order}\geq 2$) stationary flows of the KdV hierarchy.
\end{example}

\section{Application to the Witten--Kontsevich tau function}
\label{sec:WK}
As we have already mentioned in the introduction, the point of the Sato Grassmannian for the Witten-Kontsevich tau function $Z(\bt)$ is given by
\be
W^{WK} = \Span_\mathbb{C}\left\{\lambda^{2k}c(\lambda), \lambda^{2k+1}q(\lambda)\right\}_{k=0}^{\infty}.
\ee
which satisfies
\be \lambda^2\,W^{WK}\subset\,W^{WK}.\ee
Here $c(\lambda)$ and $q(\lambda)$ are defined in \eqref{def-c} and \eqref{def-q}, respectively.

By Prop. \ref{aff-tau}
we have 
\be
Z(\bt)=\sum_{\mu\in\mathbb{Y}} A_\mu s_\mu(\btheta).
\ee
Here $A_\mu$ is given by the Giambelli-type formula \eqref{eq:Giambelli} with $A_{i,j}$ the affine coordinates associated to the Witten--Kontsevich subspace $W^{WK},$
\be
t_k = -(2k+1)!!\,\theta_{2k+1}, \qquad k\geq 0,
\ee
and $s_\mu(\btheta)$ is the Schur polynomial in $\btheta$ associated to the partition $\mu.$
 
In the rest of this section, we derive the explicit expressions of $A_{m,n}$ for $Z(\bt)$ and prove Theorem \eqref{main}. The loop group element
associated to the Witten--Kontsevich subspace $W^{WK}$ is
\be
G(\lambda) =  \frac{1}{2}\lambda^{-\frac{1}{4}\sigma_3}
\begin{pmatrix}
1 & 1\\
1 & -1
\end{pmatrix}
\begin{pmatrix}
c(\sqrt{\lambda}) & q(\sqrt{\lambda})\\
c(-\sqrt{\lambda}) & q(-\sqrt{\lambda})
\end{pmatrix}
\lambda^{\frac{1}{4}\sigma_3}.
\ee
Recall that $c(\lambda),q(\lambda)$ satisfy the following identity \cite{Itzykson_Zuber}:
\begin{align}
c(\lambda)q(-\lambda)+c(-\lambda) q(\lambda) \equiv 2,
\end{align}
and hence $\det(G(\lambda)) =1$. Moreover, since both $c(\lambda)$ and $q(\lambda)$ are power
series in $\lambda^{-3}$, the following identity holds:
\be
G\left(\omega^2 \lambda\right) = \omega^{\sigma_3}\,G(\lambda)\, \omega^{-\sigma_3}
\ee
where $\omega:=e^{\pi i/3}$. We have
\begin{align}
G_{3j} =\begin{pmatrix}
c_{2j} & 0\\
0 & q_{2j}
\end{pmatrix}, &\qquad
U_{3j} =\begin{pmatrix}
q_{2j} & 0\\
0 & c_{2j}
\end{pmatrix},\\
G_{3j+1} =\begin{pmatrix}
0 & q_{2j+1}\\
0 & 0
\end{pmatrix}, &\qquad
U_{3j+1} =\begin{pmatrix}
0 & -q_{2j+1}\\
0 & 0
\end{pmatrix},\\
G_{3j+2} =\begin{pmatrix}
0 & 0\\
c_{2j+1} & 0
\end{pmatrix}, &\qquad
U_{3j+2} =\begin{pmatrix}
0 & 0\\
-c_{2j+1}& 0
\end{pmatrix}.
\end{align}

Substituting the above expressions of $G_k,\,U_k$ into equation \eqref{eq:Z} we
can find the expressions for $Z_{k,l}$. 
For example,
\begin{align}
Z_{3k,3l}&= -\sum_{j=0}^{3k}G_{j}U_{3k+3l+1-j}\nn\\
&= -\sum_{j_1=0}^{k}G_{3j_1}U_{3k+3l+1-3j_1} - \sum_{j_2=0}^{k-1}G_{3j_2+1}U_{3k+3l-3j_2}\nn\\
&= 
\begin{pmatrix}
0 & \sum_{j_1=0}^{k}c_{2j_1}q_{2k+2l-2j_1+1} -\sum_{j_2=0}^{k-1}q_{2j_2+1}c_{2k+2l-2j_2}\\
0 & 0
\end{pmatrix}.
\end{align}

Now let us identify the affine coordinates $A_{m,n}$ with the numbers $A^{Zhou}_{m,n}$ derived by Zhou \cite{Zhou}. To do so we rescale $A^{Zhou}_{m,n}$ as in the introduction
\be B_{m,n}=(\sqrt{-2})^{m+n+1}A^{Zhou}_{m,n},\quad m,n\geq 0.\ee
\begin{proposition} \label{id-Zhou} For all $m,n\geq 0,$
\be A_{m,n}=B_{m,n}. \ee
\end{proposition}
\begin{proof} Due to Lemma \ref{unique-Z} it suffices to show that
\begin{align}
&B_{k+2,l}-B_{k,2+l} = B_{k,0}B_{1,l}+B_{k,1}B_{0,l},\quad \forall\,k,l\geq 0,\label{to-show1}\\
&B_{k,0}=A_{k,0},\quad B_{k,1}=A_{k,1},\quad \forall\,k\geq 0,\label{to-show2}\\
&B_{0,k}=A_{0,k},\quad B_{1,k}=A_{1,k}.\quad \forall\,k\geq 0.\label{to-show3}
\end{align}
Noting that $B_{m,n}=0$ unless $m+n\equiv-1\,(\text{mod} ~3)$ and that \be B_{3m-1,3n}=\,B_{3m-3,3n+2},\ee in order to show \eqref{to-show1} we only need to show
\be B_{3m-2,3n+1}-B_{3m,3n-1}=\,-B_{3m-2,1}B_{0,3n-1}.\ee
This is equivalent to the following
\begin{align}\label{combinatorics}
&m(2m+1)\left(B_n(m)+\frac{b_n}{6m-1}\right)\nn\\
&-(6m+7)(6m+5)(6m+3)\left(B_{n-1}(m+1)+\frac{b_{n-1}}{6m+7}\right)\nn\\
&\qquad =\frac{(6n-1)!!}{6m-1}\frac{(2m+2n)(2m+2)!}{(2m)!(2n)!}\frac{1}{(m+1)...(m+n-1)}.
\end{align}
Substituting the recursion relation 
\begin{multline}
B_n(x)=108(x+2)\cdot B_{n-1}(x+1)\\+105\frac{B_{n-1}(x+1)}{x}-\frac{18(n-1)b_{n-1}}{x}+18b_{n-1}
\end{multline}
into Eq.~\eqref{combinatorics} above, we find that \eqref{combinatorics} becomes a simple combinatorial identity.

Noticing that both $B_{m,n}$ and $A_{m,n}$ satisfy the following symmetries
\be B_{n,m}=(-1)^{m+n}B_{m,n},\qquad A_{n,m}=(-1)^{m+n}A_{m,n}, \ee
we find \eqref{to-show2} and \eqref{to-show3} are equivalent. It remains to show \eqref{to-show2}. And this is true because of \eqref{ini-Z}. The proposition is proved.
\end{proof}

It is clear that Prop. \ref{id-Zhou}, Lemma \ref{m-aff-1}, Lemma \ref{m-aff-2}, Lemma \ref{m-symm} together imply Theorem \ref{main}.
\section{Relation to the $R$-matrix of the $A_2$ Frobenius manifold}
\label{sec:A2_Frobenius}
The Frobenius manifold corresponding to moduli spaces of $3$-spin structures of type $A$ \cite{Witten2} is the space of miniversal deformations of a simple singularity of type $A_2.$
It is usually called the $A_2$ Frobenius manifold. The {\em potential} of this Frobenius manifold \cite{Dubrovin,PPZ} is given by
\be F=\frac{1}{2}(v^1)^2v^2+\frac{1}{72}(v^2)^4,\ee
with the invariant metric
\be\eta=\begin{pmatrix}
0 & 1\\
1 & 0\\ \end{pmatrix}.
\ee
Here $v=(v^1,v^2)$ is a flat coordinate system at certain semisimple point of the Frobenius manifold. The spectral data $(\mu,\rho)$ \cite{Dubrovin} for this Frobenius manifold reads as follows
\be\mu_1=-1/6,\,\mu_2=1/6,\,\rho=0.\ee

The $R$-matrix of a Frobenius manifold was introduced in \cite{Dubrovin,Givental}.
According to an explicit expression of $R(z)$ given in \cite{PPZ}, we have
\be
R(z)=\begin{pmatrix}
\sum_k q_{2k}z^{2k} & -\sum_k q_{2k+1} z^{2k+1}\\
-\sum_k c_{2k+1}z^{2k+1} & \sum_k c_{2k} z^{2k}\\
\end{pmatrix}.
\ee
As before $\sum_k:=\sum_{k=0}^\infty.$ Note that here we have evaluated the $R$-matrix $R(v;z)$ of the $A_2$ Frobenius manifold at a particular point of this manifold, more precisely at $\phi = 6^{-2/3}$ in the notation of \cite{PPZ}.

\begin{prf} of Theorem \ref{two-affine}.\quad It is straightforward to verify that
\be \label{relation-RG}
R(z) = z^{\frac{\sigma_3}{6}}G^{-1}\left(z^{-\frac{2}{3}}\right)z^{-\frac{\sigma_3}{6}}.
\ee
Recall that the $V$-matrices $V_{k,l}$ associated to $R(z)$ are defined by
\be \label{V-def}
\frac{R^*(w)R(z)-I}{w+z}=\sum_{k,l\geq 0} (-1)^{k+l}\, V_{k,l}\, w^k z^l.
\ee
Here $R^{*}(w) = \eta \,R^{T}(w)\,\eta$. These matrices satisfy that $\forall\,k,l\geq 0,$
\bea
V_{k,l+1}+V_{k+1,l}&=&V_{k,0}\,V_{0,l},\\
V_{k,l}^*&=&V_{l,k}.
\eea  
Proof of \eqref{a1},\eqref{a2} is then straightforward by using \eqref{relation-RG}. \end{prf}

\section{Conclusion}
\label{sec:conclusion}
Ityzkson and Zuber \cite{Itzykson_Zuber} also considered more general matrix integrals, which give the partition function for the $r$-spin structures of type $A,$ for any $r\geq 2$. See also Kac-Schwarz \cite{KS} by using the Grassmannian approach. Note that our formulation for matrix-valued affine coordinates and tau functions \eqref{tau-kdv} also work for any Gelfand-Dickey hierarchy. We will postpone explicit formulas for the $r$-spin partition function in a subsequent paper \cite{BY}.

Another interesting question is to investigate the relation between the Taylor coefficients $(s_0,s_1,s_2,...)$ of an analytic initial data $u_0(x)$ of the KdV hierarchy
\be
u_{t_{\geq 1}=0}=u_0(x)=\sum_{n=0}^\infty \frac{s_n}{n!}\,x^n
\ee
and the matrix-valued affine coordinates (or alternatively speaking the $G-$matrix \eqref{G-Z}). For example, our results imply that the $G-$matrix
\be
G(\lambda) = \begin{pmatrix}1& c\,\lambda^{-1}\\ 0& 1\end{pmatrix}
\ee
gives rise to the following choice of coefficients
\be s_0=0,\,s_1=2c,\,s_2=s_3=0,\,s_4=\frac53c^2,\,s_5=s_6=0,\,s_7=\frac89c^3,\,\dots,\ee
and the $G-$matrix 
\be
G(\lambda)=\left(\begin{array}{ll}
\sum_k c_{2k}\lambda^{-3k} &  \sum_k q_{2k+1} \lambda^{-3k-1}\\
\sum_k c_{2k+1}\lambda^{-3k-2} & \sum_k q_{2k} \lambda^{-3k}\\
\end{array}\right).
\ee
gives rise to the coefficients
\be
s_0=0,\,s_1=1,\,s_2=s_3=s_4=...=0.
\ee
Considering Theorem \ref{two-affine}, it would be interesting to investigate the initial data $u_0(x)$ corresponding to the $G$-matrix coming from the $R-$matrix associated to an arbitrary semisimple calibrated two dimensional Frobenius manifold. 
\subsection*{Acknowledgements}
We would like to thank Boris Dubrovin and Marco Bertola for many helpful discussions and encouragements. F. B. wishes to thank John Harnad for introducing him to the subject of tau functions. D. Y. wishes to thank Youjin Zhang for his advises and helpful discussions. The work is partially supported by PRIN 2010-11 Grant ``Geometric and analytic theory of Hamiltonian systems in finite and infinite dimensions" of the Italian Ministry of Universities and Researches, and by the Marie Curie IRSES project RIMMP.


\begin{thebibliography}{99}
\bibitem{BY}
Balogh, F., Yang, D., Zhou, J. in preparation.

\bibitem{Cafasso}
Cafasso, M. (2008). Block {T}oeplitz determinants, constrained {KP} and {G}elfand-{D}ickey hierarchies,
{\em Math. Phys. Anal. Geom.}, \textbf{11} (1): 11--51.


\bibitem{Dubrovin} Dubrovin, B. (1996). Geometry of 2D topological field theories. In ``Integrable Systems and
Quantum Groups", Editors: Francaviglia, M., Greco, S. {\em Springer Lecture Notes in Math.} \textbf{1620}: 120--348.

\bibitem{HE}
\`Enolskii, V. Z., Harnad, J. (2011). Schur function expansions of KP $\tau$-functions associated to algebraic curves. {\em Russian Mathematical Surveys}, {\bf 66}(4), 767.

\bibitem{Gantmacher} Gantmacher, F. R. (2000). The Theory of Matrices, vols. I and II. {\em AMS Chelsea Publishing, Providence R.I.} (reprinted)

\bibitem{Givental}
Givental, A. (2001). Gromov-Witten invariants and quantization of quadratic Hamiltonians. {\em Mosc. Math. J.,} {\bf 1} (4): 551-568.

\bibitem{Itzykson_Zuber}
Itzykson, C., Zuber, J.-B. (1992). Combinatorics of the modular group. {II}. The Kontsevich
  integrals.  {\em Internat. J. Modern Phys. A}, {\bf 7} (23): 5661--5705.

\bibitem{KS}
Kac, V., Schwarz, A. (1991). Geometric interpretation of the partition function of 2D gravity. {\em Physics letters B}, {\bf 257} (3):  329--334.

\bibitem{Kontsevich}
Kontsevich, M. (1992). Intersection theory on the moduli space of curves and the matrix
  Airy function.  {\em Comm. Math. Phys.}, {\bf 147} (1): 1--23.
  
\bibitem{Macdonald}
Macdonald, I. G. (1995). Symmetric functions and Hall polynomials. Second Edition. {\em Oxford Mathematical Monographs}. Oxford University Press Inc., NewYork.


\bibitem{PPZ}
Pandharipande, R., Pixton, A., Zvonkine, D. (2013). Relations on $\overline{\mathcal{M}}_ {g, n}$ via 3-spin structures. {\em arXiv preprint arXiv:} 1303.1043.

\bibitem{Sato}
Sato, M. (1981). Soliton Equations as Dynamical Systems on a Infinite Dimensional Grassmann Manifolds (Random Systems and Dynamical Systems). {\em RIMS Kokyuroku}, {\bf 439}: 30--46.

\bibitem{Segal_Wilson}
Segal, G., Wilson, G. (1985). Loop groups and equations of KdV type. {\em Inst. Hautes \'Etudes Sci. Publ. Math.}, {\bf 61}: 5--65.

\bibitem{Witten}
Witten, E. (1991). Two-dimensional gravity and intersection theory on moduli space. {\em Surveys in Diff. Geom.}, {\bf 1}: 243--310.

\bibitem{Witten2}
Witten, E. (1993). Algebraic geometry associated with matrix models of two-dimensional gravity (pp. 235-269). {\em Topological methods in modern mathematics} (Stony Brook, NY, 1991), Publish or Perish, Houston, TX.

\bibitem{Zhou}
Zhou, J. (2013). Explicit Formula for Witten-Kontsevich Tau-Function. {\em arXiv preprint arXiv:} 1306.5429.
\end{thebibliography}
\end{document}